\documentclass[11pt]{article}

\usepackage[utf8]{inputenc}
\usepackage[T1]{fontenc}
\usepackage{amsmath,amssymb,amsthm}
\usepackage{graphicx}
\usepackage{booktabs}
\usepackage{authblk}
\usepackage[margin=1in]{geometry}
\usepackage{hyperref}
\hypersetup{colorlinks=true, linkcolor=blue, citecolor=blue, urlcolor=blue}

\newtheorem{lemma}{Lemma}
\newtheorem{proposition}{Proposition}

\title{When a common price signal is present, network topology\\
leaves no fingerprint on a storage fleet's collective dynamics}

\author[1]{Stelios Savva}
\affil[1]{Independent researcher, Cyprus \\ \texttt{savvasta2@gmail.com}}
\date{\today}

\begin{document}
\maketitle

\begin{abstract}
Price based mean field models of battery storage coordination assume each agent
responds to the true population average charging power. Under that assumption a
unique, stable equilibrium exists and the communication topology between agents
is, unsurprisingly, irrelevant: a broadcast price tells every agent what it
needs to know. We ask what happens one step away from that assumption, when
agents act instead on a shared but \emph{noisy} forecast of the average, with
correlation $\rho$ between agents' forecast errors. We show, analytically and in simulation, that topology remains undetectable across the entire regime,
including the limit where the broadcast channel is switched off and neighbour
observation is an agent's only signal. The reason is structural: the correlated
part of the forecast error projects entirely onto the one collective mode that
no graph can alter, so the effective dimensionality of the fleet is set by
$\rho$ and the fleet size, never by the topology. This is the regime the
mean field game literature sets aside by construction; the note characterises
it and gives operators a concrete thing to check (whether their fleets share
forecast providers) rather than a network to build.
\end{abstract}

\section{Where this sits}

Decentralised charging of large storage populations is usually coordinated
through a price that rises with aggregate demand. Al Dandachly, Gao and
Malham\'e~\cite{aldandachly} give a clean mean field game (MFG) treatment: with
charging power as a state and its ramp rate as the control, a price monotone in
the population average charging power yields a unique equilibrium for any time
horizon, with no contraction condition. In that framework each agent reacts to
the \emph{true} mean field $\bar x_2(t)$; the consistency condition
$\bar x_2 = \mathbb{E}[x_{2,i}]$ builds it in. Whether the agents are wired in a
line, a star, or a random graph makes no difference, because a broadcast price
already carries the only coupling that matters.

Assuming agents see the true mean field is the right move for that work, and a
deliberate one: it is what makes the equilibrium tractable, and nothing in the
equilibrium question requires more. We are not pointing at a gap they missed. We
are pointing at a knob they had every reason to leave fixed, and turning it,
because what happens when you turn it is simply interesting.

The knob is this. In practice agents do not observe the true average; they
observe a forecast of it, and if they buy that forecast from the same few
providers their forecast errors are correlated. Write the error of agent $i$ as
\begin{equation}
\varepsilon_i(t) = \sqrt{\rho}\,Z(t) + \sqrt{1-\rho}\,\eta_i(t),
\qquad \mathrm{Corr}(\varepsilon_i,\varepsilon_j)=\rho ,
\label{eq:decomp}
\end{equation}
where $Z$ is a shock common to all agents and $\eta_i$ is private. Note that
$\rho$ is not something an operator sets; it emerges on its own from shared
information sources, which is part of why it is worth looking at. At $\rho=0$ the
errors wash out in the average and the MFG picture is recovered exactly, the
fixed knob sitting at zero. The question here is what the rest of the dial does, and
in particular whether the network topology, irrelevant under a clean common
signal, wakes up once that signal is corrupted.

The short answer is no, and for a reason that turns out to be exact.

\section{Model}

We simulate $N$ batteries over a day of Cyprus style PV generation. Each agent
runs a time of use rule modulated by a hybrid controller that blends two
signals with a mixing weight $\alpha\in[0,1]$:
\begin{equation}
\text{action scale} = \alpha\, s_{\text{common}} + (1-\alpha)\, s_{\text{neigh}},
\label{eq:mix}
\end{equation}
where $s_{\text{common}}$ is driven by the (noisy) perceived feeder/price signal
and $s_{\text{neigh}}$ by the average action of graph neighbours. Topology enters
only through $s_{\text{neigh}}$. Setting $\alpha=1$ trusts the common signal
alone; $\alpha=0$ closes the common channel and leaves the agent with neighbour
observation only. We test three topologies, a linear chain, a star, and a
Watts Strogatz small world graph, and sweep $\rho$ over $[0,1]$.

One elementary fact drives everything and is worth stating on its own. The
population mean forecast error has variance
\begin{equation}
\mathrm{Var}(\bar\varepsilon)
= \frac{\sigma^2}{N}\bigl[1+(N-1)\rho\bigr]
\;\xrightarrow[N\to\infty]{}\; \rho\,\sigma^2 .
\label{eq:var}
\end{equation}
At $\rho=0$ this is $O(1/N)$ and vanishes with fleet size, the law of large
numbers, the thing the MFG equilibrium quietly relies on. At any $\rho>0$ it
sits at $O(1)$ and does not vanish, no matter how large the fleet. Correlated
error does not average away. Everything below is a consequence of where that
$O(1)$ term goes.

\section{A dimensionality diagnostic}

We measure the fleet's collective behaviour without going through the
controller's own notion of synchrony, to keep the measurement independent of
the mechanism that produced it. Stack the per agent charging power into a matrix
$X\in\mathbb{R}^{T\times N}$ and take its singular values $s_k$. Define
\begin{equation}
w_1=\frac{s_1^2}{\sum_k s_k^2},
\qquad
\mathrm{PR}=\frac{\bigl(\sum_k s_k^2\bigr)^2}{\sum_k s_k^4}.
\end{equation}
$w_1$ is the fraction of variance carried by the single dominant mode; $\mathrm{PR}$
is the participation ratio, an effective count of active collective dimensions.
$w_1\to1$ and $\mathrm{PR}\to1$ mean the fleet has collapsed onto one mode; a
spread of energy across modes ($w_1$ low, $\mathrm{PR}$ high) is what an active
topology would look like.

\section{Why topology cannot show up: the argument}

In the linearised regime the result is not empirical. It is forced. Linearise
the fleet about an operating point:
\begin{equation}
x(t+1) = a\,x(t) + \kappa\,L\,x(t) + b\,\varepsilon(t),
\label{eq:lin}
\end{equation}
with $a\in(-1,1)$ the self dynamics, $\kappa$ the neighbour coupling, and $L$ the
graph Laplacian, the only place topology appears. The forecast error has
covariance $\mathrm{Cov}(\varepsilon)=\rho\sigma^2\mathbf{1}\mathbf{1}^\top+(1-\rho)\sigma^2 I$
from \eqref{eq:decomp}, where $\mathbf{1}$ is the all ones vector.

Because $L$ is symmetric it has an orthonormal eigenbasis
$L=\sum_m\lambda_m v_m v_m^\top$, $0=\lambda_1\le\lambda_2\le\dots$. Projecting
\eqref{eq:lin} onto mode $m$ ($\hat x_m=v_m^\top x$) decouples it into $N$
independent scalar recursions
$\hat x_m(t+1)=(a+\kappa\lambda_m)\hat x_m(t)+b\,\hat\varepsilon_m(t)$, each with
stationary variance
$\mathrm{Var}(\hat x_m)=b^2\,\mathrm{Var}(\hat\varepsilon_m)/(1-(a+\kappa\lambda_m)^2)$.
Topology sits entirely in the gains, through $\lambda_m$. The question is where
the driving variance $\mathrm{Var}(\hat\varepsilon_m)$ lands.

\begin{lemma}
\label{lem}
For any connected graph,
\begin{equation}
\mathrm{Var}(\hat\varepsilon_m)=\rho\sigma^2(v_m^\top\mathbf{1})^2+(1-\rho)\sigma^2
=\begin{cases}
\rho\sigma^2 N+(1-\rho)\sigma^2, & m=1,\\
(1-\rho)\sigma^2, & m\ge2,
\end{cases}
\end{equation}
independently of the graph. The correlated part of the forecast error projects
entirely onto the consensus mode $v_1=\mathbf{1}/\sqrt N$ and drives no other
mode.
\end{lemma}

\begin{proof}
$\mathrm{Var}(\hat\varepsilon_m)=v_m^\top\mathrm{Cov}(\varepsilon)v_m
=\rho\sigma^2(v_m^\top\mathbf{1})^2+(1-\rho)\sigma^2$, using $\|v_m\|=1$. For a
connected graph $L\mathbf{1}=0$, so $\mathbf{1}$ is the eigenvector with
$\lambda_1=0$; normalising, $v_1=\mathbf{1}/\sqrt N$ gives $v_1^\top\mathbf{1}=\sqrt N$,
and orthogonality gives $v_m^\top\mathbf{1}=0$ for $m\ge2$.
\end{proof}

The consequence is immediate. The energy ratio of the consensus mode to any
other is
\begin{equation}
\frac{\mathrm{Var}(\hat x_1)}{\mathrm{Var}(\hat x_m)}
=\underbrace{\frac{\rho N+(1-\rho)}{1-\rho}}_{O(\rho N)}
\cdot
\underbrace{\frac{1-(a+\kappa\lambda_m)^2}{1-a^2}}_{O(1)},
\qquad m\ge2 .
\label{eq:ratio}
\end{equation}
The first factor grows like $\rho N$. The second, the only place the topology
lives, is bounded, since the Laplacian spectrum sits in a fixed interval. So as
$\rho N$ grows, $w_1\to1$ and $\mathrm{PR}\to1$ regardless of $\{\lambda_m\}$.

\begin{proposition}
In the linearised fleet, the effective dimensionality of the stationary
trajectory is set by $\rho N$ and is asymptotically independent of topology; any
topological contribution to $w_1$ is $O(1/(\rho N))$ relative to the common
signal.
\end{proposition}

The picture is that the common signal saturates the one mode no graph can move,
while topology only ever governs the modes $m\ge2$, which carry vanishing
energy. The two live in orthogonal subspaces; they cannot trade. This is not a
statement that topology does nothing (\eqref{eq:ratio} is a ratio, not a zero)
but that its share shrinks like $1/(\rho N)$ and is gone before it can be
seen. The argument assumes linear dynamics; away from controller saturation this
is fine, and the one place the experiments show extra sensitivity is exactly
near the collapse threshold, where saturation begins.

\section{Simulation}

Figure~\ref{fig:w1} shows $w_1(\rho)$ for the three topologies. The curves are
on top of each other and rise smoothly from $\approx0.44$ to $\approx0.96$; the
crossover is gradual, not a phase transition, and $\mathrm{PR}$ falls in step
from $\approx4.3$ to $\approx1.1$. Reading the three topologies apart by eye is
not possible.

\begin{figure}[t]
  \centering
  \includegraphics[width=0.62\linewidth]{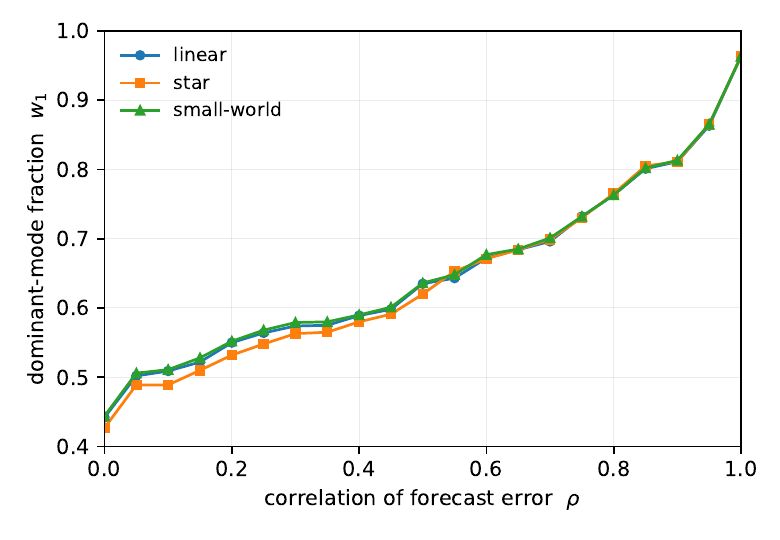}
  \caption{Dominant mode fraction $w_1$ against forecast error correlation
  $\rho$, at mixing $\alpha=0.5$. The three topologies are indistinguishable and
  the fleet collapses smoothly toward a single mode as $\rho$ grows.}
  \label{fig:w1}
\end{figure}

To confirm the Lemma directly, Figure~\ref{fig:cos} plots the alignment between
the measured dominant mode $u_1$ and the consensus vector $\mathbf{1}$. At
$\rho=1$ the alignment is $0.998$ for all three topologies. The dominant mode
\emph{is} the consensus vector, as predicted. Even at $\rho=0$ it is already
$0.96$, because the broadcast price signal at $\alpha=0.5$ is itself pushing the
fleet toward consensus; raising $\rho$ tightens it the rest of the way.

\begin{figure}[t]
  \centering
  \includegraphics[width=0.62\linewidth]{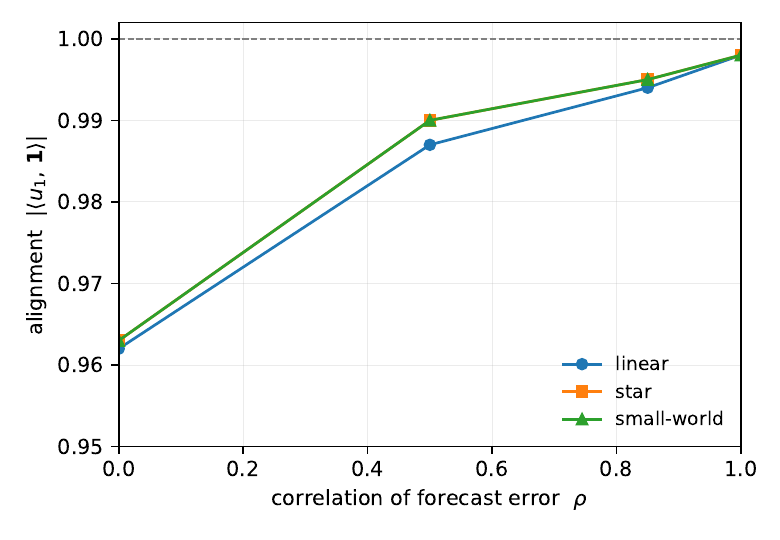}
  \caption{Alignment of the empirical dominant mode $u_1$ with the consensus
  vector $\mathbf{1}$. All topologies converge to $\approx0.998$ at $\rho=1$,
  the direct confirmation of Lemma~\ref{lem}.}
  \label{fig:cos}
\end{figure}

\paragraph{Is the small difference real or noise?}
The between topology spread in $w_1$ is tiny but nonzero, so it is worth asking
whether it is signal at all. We ran $30$ noise seeds per setting with a
bootstrap confidence interval, then built a null: three copies of the
\emph{same} linear graph, each fed an independent noise seed, so any spread
between them is pure noise realisation variability at fixed topology.
Table~\ref{tab:null} gives the comparison. The real between topology spread
($0.004$ at $\alpha=0.5$, $0.018$ at $\alpha=0$) is smaller than the spread
noise alone produces on a fixed graph ($\approx0.03$). Re drawing the forecast
noise moves $w_1$ more than switching topology does. This holds at both ends of
the mixing axis, including $\alpha=0$ where the common channel is fully closed.

\begin{table}[t]
  \centering
  \begin{tabular}{lccc}
    \toprule
    $\alpha$ & between topology spread & same topology noise spread & \\
    \midrule
    $0.5$ & $0.0043$ & $0.029$ \ (CI $[0.010,0.058]$) & below noise \\
    $0.0$ & $0.0181$ & $0.031$ \ (CI $[0.011,0.061]$) & below noise \\
    \bottomrule
  \end{tabular}
  \caption{Topology spread versus the noise floor, $N=30$ seeds. The topological
  difference is smaller than the variability from re drawing the forecast noise
  on one fixed graph, at both extremes of the mixing weight.}
  \label{tab:null}
\end{table}

\section{What this is and is not}

The claim is narrow and worth stating plainly. Under a shared PV derived price
signal, the communication topology has no detectable effect on the fleet's
collective dimensionality, across the whole mixing axis and down to the limit
where topology is the only channel left. The effect is not merely small at the
$\alpha=0.5$ operating point one might test by default; it is below the noise
floor everywhere, and the Lemma says why: the correlated signal and the graph
occupy orthogonal subspaces.

A few honest boundaries. This is detectability at $30$ seeds, not a proof of
exact zero; more seeds could surface a small effect, though \eqref{eq:ratio}
bounds how small. The noise sensitivity peaks near the collapse threshold
$\rho\approx0.8$, a signature of criticality and a plausible early warning
direction we have not pursued here. And the quantity measured is dimensionality,
not the collapse threshold itself; topology could in principle nudge the exact
threshold without changing dimensionality.

The framing we would resist is that this is a design recommendation. It is not a
result that says ``do not build a mesh network'' from field data. It is a
mechanism: it says that when a common signal is present, fleet size and forecast
correlation set the collective behaviour and the wiring does not, so the thing
worth checking in a real deployment is not the communication graph but whether
the fleet's forecasts come from the same source. That is the lever
\eqref{eq:var} points at: reduce the correlated error, by better forecasting
or by decorrelating providers, and the common mode loses its grip; rewire the
graph and nothing moves.

Finally, this is deliberately the regime the MFG treatment sets aside. Its
equilibrium is built on agents seeing the true mean field, i.e.\ $\rho=0$; the
correlated belief axis is orthogonal to what that construction can see. The note
does not overturn it; it maps the one direction it leaves out.

\section*{Reproducibility}
Code and experiments: \url{https://github.com/rrumabo/semele}.


\begin{thebibliography}{9}
\bibitem{aldandachly}
N.~Al Dandachly, S.~Gao, R.~Malham\'e,
``Price Coordinated Mean Field Games with State Augmentation for Decentralized
Battery Charging,'' arXiv:2604.05269, 2026.
\bibitem{lumaggioni}
F.~Lu, M.~Maggioni \emph{et al.},
``Data driven Discovery of Emergent Behaviors in Collective Dynamics,''
arXiv:1912.11123.
\bibitem{dacunha}
G.~S.~Y. Giardini, J.~F. Hardy II, C.~R. da Cunha,
``Evolving Neural Networks Reveal Emergent Collective Behavior from Minimal
Agent Interactions,'' arXiv:2410.19718, 2024.
\end{thebibliography}
\end{document}